\documentclass[a4paper, 12pt, reqno]{amsart}

\usepackage{amsmath,amsfonts,amsthm,amssymb,here}
\usepackage{hyperref}
\usepackage{amsmath}
\usepackage{graphicx}
\begin{document}
\bibliographystyle{plain}
\newtheorem{theo}{Theorem}[section]
\newtheorem{lemme}[theo]{Lemma}
\newtheorem{cor}[theo]{Corollary}
\newtheorem{defi}[theo]{Definition}
\newtheorem{prop}{Proposition}
\newtheorem{problem}[theo]{Problem}
\newtheorem{remarque}[theo]{Remark}
\newtheorem{claim}[theo]{Claim}
\newcommand{\beq}{\begin{eqnarray}}
\newcommand{\enq}{\end{eqnarray}}
\newcommand{\be}{\begin{eqnarray*}}
\newcommand{\en}{\end{eqnarray*}}
\newcommand{\ben}{\begin{eqnarray*}}
\newcommand{\enn}{\end{eqnarray*}}
\newcommand{\Td}{\mathbb T^d}
\newcommand{\Rd}{\mathbb R^n}
\newcommand{\R}{\mathbb R}
\newcommand{\N}{\mathbb N}
\newcommand{\Sn}{\mathbb S}
\newcommand{\Zd}{\mathbb Z^d}
\newcommand{\Linf}{L^{\infty}}
\newcommand{\dt}{\partial_t}
\newcommand{\Dt}{\frac{d}{dt}}
\newcommand{\Dtt}{\frac{d^2}{dt^2}}
\newcommand{\demi}{\frac{1}{2}}
\newcommand{\vf}{\varphi}
\newcommand{\epu}{_{\epsilon}}
\newcommand{\ep}{^{\epsilon}}
\newcommand{\bfi}{{\mathbf \Phi}}
\newcommand{\bpsi}{{\mathbf \Psi}}
\newcommand{\bx}{{\mathbf x}}
\newcommand{\bX}{{\mathbf X}}
\newcommand{\ds}{\displaystyle}
\newcommand {\g}{\`}
\newcommand{\E}{\mathbb E}
\newcommand{\Q}{\mathbb Q}
\newcommand{\PP}{\mathbb P}
\newcommand{\1}{\mathbb I}
\let\cal=\mathcal

\title{Performance analysis of the optimal strategy under partial information}
\maketitle
\begin{center}
Ahmed Bel Hadj Ayed\footnotemark[1]$^{,}$\footnotemark[2], Gr\'egoire Loeper \footnotemark[2], Sofiene El Aoud \footnotemark[1],  Fr\'ed\'eric Abergel \footnotemark[1] 
\end{center}
\begin{abstract}
The question addressed in this paper is the performance of the optimal strategy, and the impact of partial information. The setting we consider is that of a stochastic asset price model where the trend follows an unobservable Ornstein-Uhlenbeck process. We focus on the optimal strategy with a logarithmic utility function under full or partial information. For both cases, we provide the asymptotic expectation and variance of the logarithmic return as functions of the signal-to-noise ratio and of the trend mean reversion speed. Finally, we compare the asymptotic Sharpe ratios of these strategies in order to quantify the loss of performance due to partial information. 
\end{abstract}
\footnotetext[1]{Chaire of quantitative finance, laboratory MAS, CentraleSup\'elec}
\footnotetext[2]{BNP Paribas Global Markets}

\section*{Introduction}
Optimal investment was introduced by Merton in 1969 (see \cite{Merton} for details). He assumed that the risky asset follows a geometric Brownian motion and derived the optimal investment rules for an investor maximizing his expected utility function. Several generalisations of this problem are possible. One of them is to consider a stochastic unobservable trend, which leads to a system with partial information. This hypothesis seems to be realistic since only the historical prices of the risky asset are available to the public. 
For example,  Karatzas and Zhao (see \cite{Karatzas}) study the case of an unobservable constant trend, Lakner (see \cite{LaknerStrat}) and Brendle (see \cite{Brendle}) consider a stochastic asset price model where the trend is an unobservable Ornstein Uhlenbeck process, and Sass and Hausmann (see \cite{Sass}) suppose that the trend is given by an unobserved continuous time, finite state Markov chain.

In this paper, we consider a stochastic asset price model where the trend is an unobservable Ornstein Uhlenbeck process and we focus on the optimal strategy with a logarithmic utility function  under partial or complete information.

The purpose of this work is to characterize the performance of these strategies as functions of the signal-to-noise ratio and of the trend mean reversion speed and to quantify the loss of performance due to partial information. The loss of utility due to incomplete information was already studied by 
Karatzas and Zhao (see \cite{Karatzas}), by  Brendle (see \cite{Brendle}) and by Rieder and B\"auerle (see  \cite{rieder}). Here, the trading strategy performance is measured with the asymptotic Sharpe ratio (see \cite{Sharpe} for details).

The paper is organized as follows: the first section presents the model and recalls some results from filtering theory.

In the second section, the optimal strategy with complete information is investigated. This portfolio is built by an agent who is able to observe the trend and aims to maximize his expected logarithmic utility. We provide, in closed form, the expectation and variance of the logarithmic return as functions of the signal-to-noise ratio. We also show that the asymptotic Sharpe ratio of the  optimal strategy with complete information is an increasing function of the signal-to-noise ratio.

In the third section, we consider the optimal strategy under partial information. This corresponds to an unobservable trend process and to an agent who aims to maximize his expected logarithmic utility. In this case, we provide, in closed form, the expectation and variance of the logarithmic return as functions of the signal-to-noise ratio and of the trend mean reversion  speed. Then, we derive the asymptotic Sharpe ratio and we show that this is an increasing function of the signal-to-noise ratio and an unimodal (increasing then decreasing) function of the trend mean reversion speed. After that, we introduce the partial information factor which is the ratio between the asymptotic Sharpe ratio of the optimal strategy with partial information and the asymptotic Sharpe ratio of the optimal strategy with full information. This factor measures the loss of performance due to partial information. We show that this factor is bounded by a threshold equal to $\frac{2}{3^{3/2}}$. 

In the fourth section, numerical examples illustrate the analytical results of the previous sections. The simulations show that, even with a high signal-to-noise ratio, a high trend mean reversion speed leads to a negligible performance of the optimal strategy under partial information compared to the case with complete information.


\section{Setup}
This section begins by presenting the model, which corresponds to an unobserved mean-reverting diffusion. After that, we reformulate this model in a completely observable environment (see \cite{Lipster1} for details). This setting introduces the conditional expectation of the trend, knowing the past observations. Then, we recall the asymptotic continuous time limit of the Kalman filter.
\subsection{The model}
Consider a financial market living on a stochastic basis $( \Omega , \mathcal{F} ,\mathbf{F}, \mathbb{P} )$, where  $\mathbf{F}=\left\lbrace  \mathcal{F}_{t}, t \geqslant 0 \right\rbrace$ is the natural filtration associated to a two-dimensional (uncorrelated) Wiener process $(W^{S},W^{\mu})$, and $\mathbb{P}$ is the objective probability measure. The dynamics of the risky asset $S$ is given by
\begin{eqnarray}\label{Model}
\frac{dS_{t}}{S_{t}}&=&\mu_{t}dt+\sigma_{S} dW_{t}^{S},\\
d\mu_{t}&=&-\lambda\mu_{t}dt+\sigma_{\mu}dW_{t}^{\mu}\label{Model2},
\end{eqnarray}
with $\mu_{0}=0$. We also assume that $\left( \lambda,\sigma_{\mu},\sigma_{S}\right) \in  \mathbb{R}_+^{*}\times\mathbb{R}_+^{*}\times\mathbb{R}_+^{*}$. The parameter $\lambda$ is called the trend mean reversion speed. Indeed, $\lambda$ can be seen as the "force" that pulls the trend back to zero. 
Denote by $\mathbf{F}^{S}=\left\lbrace  \mathcal{F}_{t}^S \right\rbrace$ be the natural filtration associated to the price process $S$. An important point is that only $\mathbf{F}^{S}$-adapted processes are observable, which implies that agents in this market do not observe the trend $\mu$.

\subsection{The observable framework}
As stated above, the agents can only observe the stock price process $S$. Since, the trend $\mu$ is not $F^S$-measurable, the agents do not observe it directly. Indeed, the model \hyperref[Model]{(1)-(2)} corresponds to a system with partial information.  The following proposition gives a representation of the model \hyperref[Model]{(1)-(2)} in an observable framework (see \cite{Lipster1} for details or Appendix \hyperref[sec::ProofObservableFramework]{A} for a proof). 
\begin{prop}\label{ObservableFramework}
The dynamics of the risky asset $S$ is also given by
\begin{eqnarray}\label{ObservableSDE}
\frac{dS_t}{S_t} & = & E\left[  \mu_t | \mathcal{F}^S_t\right] dt + \sigma_{S} dN_t,
\end{eqnarray}
where $N$ is a $\left( \mathbb{P} ,\mathbf{F}^{S} \right)$ Wiener process.
\end{prop}
\begin{remarque}
In the filtering theory (see \cite{Lipster1} for details), the process $N$ is called the innovation process. To understand this name, note that:
\begin{eqnarray*}
dN_t  = \frac{1}{\sigma_{S}} \left(  \frac{dS_t}{S_t} - E\left[  \mu_t | \mathcal{F}^S_t\right] dt \right) .
\end{eqnarray*}
Then, $dN_t$ represents the difference between the current observation and what we expect knowing the past observations. 
\end{remarque}
\subsection{Optimal trend estimator} 
The system \hyperref[Model]{(1)-(2)} corresponds to a Linear Gaussian Space State model (see \cite{ARMABook} for details). In this case, the Kalman filter gives the optimal estimator, which corresponds to the conditional expectation $E\left[  \mu_t | \mathcal{F}^S_t\right]$. Since $\left( \lambda,\sigma_{\mu},\sigma_{S}\right) \in  \mathbb{R}_+^{*}\times\mathbb{R}_+^{*}\times\mathbb{R}_+^{*}$, the model \hyperref[Model]{(1)-(2)} is a controllable and observable time invariant system. In this case, it is well known  that the estimation error variance converges to an unique constant value (see \cite{KalmanBucy} for details). This corresponds to the steady-state Kalman filter. The following proposition (see \cite{Ahmed} for a proof) gives a first continuous representation of the steady-state Kalman filter:
\begin{prop}
The steady-state Kalman filter has a continuous time limit depending on the asset returns:
\begin{eqnarray}\label{ContinuousEstimate}
d \widehat{\mu}_{t}=-\lambda \beta \widehat{\mu}_{t}dt+\lambda\left( \beta -1 \right) \frac{dS_{t}}{S_{t}}, 
\end{eqnarray}
where
\begin{eqnarray}\label{BetaDef}
\beta= \left( 1+ \frac{\sigma_{\mu}^{2}}{ \lambda^{2} \sigma_{S}^{2}}  \right)^{\frac{1}{2}}.
\end{eqnarray}
\end{prop}
The following proposition gives a second representation of the steady-state trend estimator $\hat{\mu}$:
\begin{prop}
Based on Equation (\ref{ContinuousEstimate}), it follows that:
\begin{eqnarray}\label{SDETrend}
d\hat{\mu}_t & = & - \lambda \hat{\mu}_t dt + \lambda \sigma_{S} \left( \beta - 1\right)  dN_t.
\end{eqnarray}
\end{prop}
\begin{proof}
Replacing $\frac{dS_{t}}{S_{t}}$ in Equations \hyperref[ContinuousEstimate]{$(4)$} by the expression of Equation \hyperref[ObservableSDE]{$(3)$}, we find Equation (\ref{SDETrend}). 
\end{proof}
\begin{remarque}
It is well known that the Kalman estimator is a Gaussian process. Here, we find that the steady-state trend estimator $\hat{\mu}$  is an Ornstein Uhlenbeck process. 
In practice, the parameters $\left( \lambda,\sigma_{\mu},\sigma_{S}\right) $ are unknown and must be estimated (see \cite{Ahmed} where the authors assess the feasibility of forecasting trends modeled by an unobserved mean-reverting diffusion). In this paper, we assume that the parameters are known.
\end{remarque}
\section{Optimal strategy under complete information}
In this section, the optimal strategy under full information is investigated. This strategy is built by an agent who is able to observe the trend $\mu$. Formally, it corresponds to the case $\mathbf{F}^{S}=\mathbf{F}$. Given this framework, we consider the optimal strategy with a logarithmic utility function. We provide, in closed form, the asymptotic expectation and variance of the logarithmic return, and the asymptotic Sharpe ratio of this strategy as functions of the signal-to-noise ratio.
\subsection{Context}
Consider the financial market defined in the first section with a risk free rate and without transaction costs.
Let $P^o$ be a self financing portfolio given by:
\begin{eqnarray*}
\frac{dP_{t}^o}{P_{t}^o}&=& \omega_t^o \frac{dS_{t}}{S_{t}},\\
P_0^o&=& x,
\end{eqnarray*}
where $\omega_t^o$ is the fraction of wealth invested in the risky asset (also named the control variable). The agent aims to maximize his expected logarithmic utility on an admissible domain $\mathcal{A}^o$ for the allocation process. In this section, we assume that the agent is able to observe the trend $\mu$. Formally, it means that $\mathcal{A}^o$ represents all the $\mathbf{F}$-progressive and measurable processes and the solution of this problem is given by:  
\begin{eqnarray*}
\omega^{*} = \arg\sup_{\omega \in \mathcal{A}^o} \mathbb{E} \left[\ln\left(P_{t}^o \right)| P_0^o=x \right].
\end{eqnarray*}
As is well known (see \cite{LaknerStrat} or \cite{OptimalStrat} for examples), the solution of this problem is given by:
\begin{eqnarray}\label{OmniscientPortfolioDynSelfFinancing}
\frac{dP_{t}^o}{P_{t}^o}&=& \frac{\mu_t}{\sigma_S^2} \frac{dS_{t}}{S_{t}},\\
P_0^o&=& x.
\end{eqnarray}
\subsection{Performance analysis of the optimal strategy under complete information}
The following proposition gives the stochastic differential equation of the portfolio $P^o$:
\begin{prop}\label{OmniscientPortfolioDyn}
Consider the portfolio $P^o$ given by Equation (\ref{OmniscientPortfolioDynSelfFinancing}). In this case,
\begin{equation}
d \ln(P_{t}^o) =  \frac{\mu_t^2}{2\sigma_{S}^{2}}dt+\frac{\mu_t}{\sigma_{S}}dW_{t}^{S}.
\end{equation}
\end{prop}
\begin{proof}
Using Equation (\ref{OmniscientPortfolioDynSelfFinancing}) and It\^o's lemma on the process $\ln(P_{t}^o)$, the result follows.
\end{proof}
The asymptotic expected logarithmic return is the first indicator to assess the potential of a trading strategy. The second one can be the variance of the logarithmic return. This indicator can be useful as a measure of risk. Moreover, let $\text{SR}_T$ be the annualized Sharpe-ratio at time $T$ of a portfolio $\left( P_T\right) $ defined by:
\begin{eqnarray}\label{SharpeRatioDef}
\text{SR}_T=\frac{\mathbb{E}\left[ \ln\left(  \frac{P_{T}}{P_{0}} \right)  \right]}{\sqrt{T\ \mathbb{V}\text{ar}\left[ \ln\left(  \frac{P_{T}}{P_{0}} \right)  \right]}}.
\end{eqnarray}
This indicator measures the expected logarithmic return per unit of risk. The Sharpe ratio is a prime metric for an investment. 
\begin{remarque}
This definition of the Sharpe ratio is different from the original one (see \cite{Sharpe}). Here, this indicator is computed on logarithmic returns. 
\end{remarque}
The following theorem gives the asymptotic expectation, variance and Sharpe ratio of the logarithmic return:
\begin{theo}\label{theoremOmniscient}
Consider the portfolio given by Equation (\ref{OmniscientPortfolioDynSelfFinancing}). In this case:
\begin{eqnarray}
\lim_{T \rightarrow \infty} \frac{\mathbb{E}\left[ \ln\left(  \frac{P_{T}^o}{P_{0}^o} \right)  \right] }{T}&=& \frac{\text{SNR}}{2}\label{OmniscientAsymptoticExp},\\
\lim_{T \rightarrow \infty} \frac{\mathbb{V}\text{ar}\left[ \ln\left(  \frac{P_{T}^o}{P_{0}^o} \right)  \right] }{T}&=& \text{SNR},\label{AsymptoticVarianceOmniscient} \\
\text{SR}_\infty^\text{o}&=&\frac{\sqrt{\text{SNR}}}{2}.\label{SharpeRatioOmniscient}
\end{eqnarray}
where $\text{SNR}$ is the signal-to-noise-ratio:
\begin{eqnarray}\label{SNR}
\text{SNR}=\frac{\sigma_{\mu}^{2}}{ 2\lambda \sigma_{S}^{2}}.
\end{eqnarray}
\end{theo}
\begin{proof}
Integrating the expression of Proposition \ref{OmniscientPortfolioDyn} from $0$ to $T$ and taking the expectation, it gives:
\begin{eqnarray*}
\mathbb{E}\left[ \ln\left(  \frac{P_{T}^o}{P_{0}^o} \right) \right]  =\frac{1}{2\sigma_{S}^{2}}\int_{0}^{T}\mathbb{E}\left[\mu_t^2\right] dt+0.
\end{eqnarray*}
Since  $\mu$ is an Ornstein-Uhlenbeck process:
\begin{eqnarray*}
\mathbb{E}\left[ \mu_t \right]  & = & 0, \\
\mathbb{V}\text{ar}\left[ \mu_t \right]  & = &\sigma_{\mu}^2 \frac{1 - e^{-2\lambda t} }{2\lambda}.
\end{eqnarray*}
Then, tending $T$ to $\infty$, Equation (\ref{OmniscientAsymptoticExp}) follows.
Since the processes $W^{S}$ and $\mu$ are supposed to be independent:
\begin{eqnarray*}
\mathbb{V}\text{ar}\left[ \ln\left(  \frac{P_{T}^o}{P_{0}^o} \right)  \right]=\frac{1}{4\sigma_S^4}\mathbb{V}\text{ar}\left[\int_{0}^{T} \mu_t^2 dt\right]+\frac{1}{\sigma_S^2}\mathbb{V}\text{ar}\left[\int_{0}^{T} \mu_t dW_{t}^{S}\right].
\end{eqnarray*}
Since the process $\left( \int_{0}^{T} \mu_t dW_{t}^{S}\right)_{T\geq 0} $ is a martingale:
\begin{eqnarray*}
\mathbb{V}\text{ar}\left[\int_{0}^{T} \mu_t dW_{t}^{S}\right]=\int_{0}^{T}\mathbb{E}\left[ \mu_t^2 \right]dt=\frac{\sigma_{\mu}^2}{2 \lambda} \left( T + \frac{1-e^{-2\lambda T}}{2\lambda}\right),
\end{eqnarray*}
Moreover, Isserlis' theorem (see \cite{Isserlis} for details) gives:
\begin{eqnarray*}
\mathbb{V}\text{ar}\left[\int_{0}^{T} \mu_t^2 dt\right]=2\int_{0}^{T}\int_{0}^{T}\left(\mathbb{E}\left[ \mu_s\mu_t\right]  \right)^2 ds dt. 
\end{eqnarray*}
Since $\mu$ is an Ornstein Uhlenbeck:
\begin{eqnarray*}
\mathbb{V}\text{ar}\left[\int_{0}^{T} \mu_t^2 dt\right]= \frac{\sigma_{\mu}^4 e^{-4\lambda T }\left( e^{4\lambda T }\left(4\lambda T-5 \right)+ e^{2\lambda T } \left(8\lambda T +4 \right) +1\right) }{8 \lambda^4}. 
\end{eqnarray*}
Equation (\ref{AsymptoticVarianceOmniscient}) follows. 
Finally, using the definition of the Sharpe ratio (see Equation (\ref{SharpeRatioDef})) and the results of  Equations (\ref{OmniscientAsymptoticExp}) and (\ref{AsymptoticVarianceOmniscient}), Equation (\ref{SharpeRatioOmniscient}) follows.
\end{proof}
Theorem \ref{theoremOmniscient} shows that the asymptotic expectation and the asymptotic variance logarithmic return are linear functions of the signal-to-noise ratio and that the asymptotic Sharpe ratio is a linear function of the ratio between the asymptotic trend standard deviation and the volatility.
\section{Optimal strategy under partial information}
In this section, the Merton's problem under partial information is investigated. We consider the case of a logarithmic utility function. We provide, in closed form, the asymptotic expectation and variance of the logarithmic return, and the asymptotic Sharpe ratio of this strategy as functions of the signal-to-noise ratio and of the trend mean reversion  speed. After that, we introduce the partial information factor which is the ratio between the asymptotic Sharpe ratio of the optimal strategy with partial information and the asymptotic Sharpe ratio of the optimal strategy with complete information. We close this section by showing that this factor is bounded by a threshold equal to $\frac{2}{3^{3/2}}$. 
\subsection{Context}
Consider the financial market defined in the first section with a risk free rate and without transaction costs.
Let $P$ be a self financing portfolio given by:
\begin{eqnarray*}
\frac{dP_{t}}{P_{t}}&=& \omega_t \frac{dS_{t}}{S_{t}},\\
P_0&=& x,
\end{eqnarray*}
where $\omega_t$ is the fraction of wealth invested in the risky asset. The agent aims to maximize his expected logarithmic utility on an admissible domain $\mathcal{A}$ for the allocation process. In this section, we assume that the agent is not able to observe the trend $\mu$. Formally, $\mathcal{A}$ represents all the $\mathbf{F}^{S}$-progressive and measurable processes and the problem is:  
\begin{eqnarray*}
\omega^{*} = \arg\sup_{\omega \in \mathcal{A}} \mathbb{E} \left[\ln\left(P_{t} \right)| P_0=x \right].
\end{eqnarray*}
The solution of this problem is well known and easy to compute (see \cite{LaknerStrat} for example). Indeed, it has the following form:
\begin{eqnarray*}
\omega^{*}_t = \frac{E\left[  \mu_t | \mathcal{F}^S_t\right]}{\sigma_{S}^{2}}.
\end{eqnarray*}
Using the steady-state Kalman filter, the optimal portfolio is given by:
\begin{eqnarray}\label{SelfFinancing}
\frac{dP_{t}}{P_{t}}&=& \frac{\widehat{\mu}_{t}}{\sigma_{S}^{2}}  \frac{dS_{t}}{S_{t}},\\
P_0&=& x,
\end{eqnarray}
where $\widehat{\mu}$ is given by Equation (\ref{ContinuousEstimate}).
\subsection{Performance analysis of the optimal strategy under partial information}
The following proposition gives the stochastic differential equation of the portfolio:
\begin{prop}\label{PortfolioDyn}
The optimal portfolio process of Equation (\ref{SelfFinancing}) follows the dynamics:  
\begin{equation*}
d \ln(P_{t}) =  \frac{1}{2\sigma_{S}^{2}\lambda\left( \beta -1 \right)} d \widehat{\mu}_{t}^{2}+\left[ \frac{\widehat{\mu}_{t}^{2}}{\sigma_{S}^{2}} \left( \frac{\beta}{\left( \beta -1 \right)}-\frac{1}{2}  \right) - \frac{1}{2} \lambda\left( \beta -1 \right) \right] dt,
\end{equation*}
where $\beta$ is given by Equation (\ref{BetaDef}).
\end{prop}
\begin{proof}
Equation  (\ref{SelfFinancing}) is equivalent to (by It\^o's lemma):
\begin{eqnarray*}
d \ln(P_{t})=  \frac{\widehat{\mu}_{t}}{\sigma_{S}^{2}}  \frac{dS_{t}}{S_{t}} - \frac{1}{2}  \frac{\widehat{\mu}_{t}^{2}}{\sigma_{S}^{2}}dt.
\end{eqnarray*}
Using Equation (\ref{ContinuousEstimate}),
\begin{align*}
d \ln(P_{t}) =  \frac{\widehat{\mu}_{t}}{\sigma_{S}^{2}} \frac{ d \widehat{\mu}_{t}}{\lambda\left( \beta -1 \right)} +  \frac{\widehat{\mu}_{t}^{2}}{\sigma_{S}^{2}} \frac{\beta}{\left( \beta -1 \right)}dt- \frac{1}{2} \frac{\widehat{\mu}_{t}^{2}}{\sigma_{S}^{2}}dt,
\end{align*}
It\^o's lemma on Equation (\ref{ContinuousEstimate}) gives:
\begin{eqnarray*}
d\widehat{\mu}_{t}^{2}=2\widehat{\mu}_{t}d\widehat{\mu}_{t}+\lambda^{2}\left( \beta_{\sigma_{\mu},\lambda,\sigma_{S}} -1 \right)^{2}\sigma_{S}^{2}dt.
\end{eqnarray*}
Using this equation, the dynamic of the logarithmic wealth follows.
\end{proof}
\begin{remarque}
Proposition \ref{PortfolioDyn} shows that the returns of the optimal strategy with partial information can be broken down into two terms. The first one represents an option on the square of the realized returns (called Option profile). The second term is called the Trading Impact. These terms are introduced and discussed in \cite{Lyxor}.
The option profile at the time $T$ is:
\begin{eqnarray*}
\text{Option Profile}_{T}=\frac{1}{2\sigma_{S}^{2}}  \frac{1}{\lambda\left( \beta -1 \right)} (\widehat{\mu}_{T}^{2}-\widehat{\mu}_{0}^{2}).
\end{eqnarray*}
With the assumption of an initial trend estimate equal to $0$, the Option profile is always positive. The Trading Impact is a cumulated function of the trend estimate:
\begin{eqnarray*}
\text{Trading impact}_{T}=\int\limits_{0}^{T} \left[ \frac{\widehat{\mu}_{t}^{2}}{\sigma_{S}^{2}} \left( \frac{\beta}{\left( \beta -1 \right)}-\frac{1}{2}  \right) - \frac{1}{2} \lambda\left( \beta -1 \right) \right] dt.
\end{eqnarray*}
When $T\rightarrow\infty$, it becomes the preponderant term.
The Trading Impact is positive on the long term $T$ if the drift estimate $\widehat{\mu}_{t}$ verifies:
\begin{eqnarray}\label{Profit}
\frac{1}{T} \int\limits_{0}^{T} \widehat{\mu}_{t}^{2}dt  >  \frac{\lambda\sigma_{S}^{2}(\beta-1)}{2\frac{\beta}{\beta-1}-1},
\end{eqnarray}
Equation (\ref{Profit}) can be seen as a condition for the trend following strategy to generate profits in the long term.
\end{remarque}
The following theorem gives the asymptotic expectation, variance and Sharpe ratio of the logarithmic return:
\begin{theo}\label{theoremPartial}
Consider the portfolio given by Equation (\ref{SelfFinancing}). In this case:
\begin{eqnarray}
\lim_{T \rightarrow \infty} \frac{\mathbb{E}\left[ \ln\left(  \frac{P_{T}}{P_{0}} \right)  \right] }{T}&=& \frac{\lambda}{4}\left( \beta-1\right)^2\label{AsymptoticExp},\\
\lim_{T \rightarrow \infty} \frac{\mathbb{V}\text{ar}\left[ \ln\left(  \frac{P_{T}}{P_{0}} \right)  \right] }{T}&=& \frac{\lambda}{8}\left( \beta^2-1\right)^2\label{AsymptoticVariance},\\
\lim_{T \rightarrow \infty} \text{SR}_T&=&\sqrt{\frac{\lambda}{2}}\frac{\beta-1}{\beta+1}\label{AsymptoticSharpeRatio},
\end{eqnarray}
where $\beta$ is given by Equation (\ref{BetaDef}).
\end{theo}
\begin{proof}
Based on Equation (\ref{SDETrend}), $\hat{\mu}$ is an Ornstein-Uhlenbeck process:
\begin{eqnarray*}
\mathbb{E}\left[ \hat{\mu}_t \right]  & = & 0, \\
\mathbb{V}\text{ar}\left[ \hat{\mu}_t \right]  & = &\left( \lambda \sigma_{S} \left( \beta-1\right)  \right) ^2 \frac{1 - e^{-2\lambda t} }{2\lambda}.
\end{eqnarray*}
Integrating the expression of Proposition \ref{PortfolioDyn} from $0$ to $T$ and taking the expectation, it gives:
\begin{align*}
& \mathbb{E}\left[ \ln\left(  \frac{P_{T}}{P_{0}} \right) \right]  = \frac{(\beta-1)}{4  } \left(  1 - e^{-2\lambda T}  \right) - \frac{\lambda (\beta-1) }{2} T  \\
& +  \left( \frac{\beta}{\beta - 1} - \frac{1}{2}\right)  \left[  \frac{\lambda}{2} (\beta-1)^2 T - \frac{\lambda}{2}(\beta-1)^2   \frac{1- e^{-2\lambda T} }{2\lambda} \right] 
\end{align*}
Then, tending $T$ to $\infty$, Equation (\ref{AsymptoticExp}) follows.
Integrating the expression of Proposition \ref{PortfolioDyn} from $0$ to $T$ and taking the variance, it gives:
\begin{eqnarray*}
\mathbb{V}\text{ar}\left[ \ln\left(  \frac{P_{T}}{P_{0}} \right)  \right]&=& \frac{1}{\left( 2\lambda\left(\beta-1 \right)\sigma_S^2\right)^2}
\mathbb{V}\text{ar}\left[ \hat{\mu}_T^2 \right]\\
&&+\frac{1}{\sigma_S^4}\left(\frac{\beta}{\beta-1}-\frac{1}{2} \right)^2\mathbb{V}\text{ar}\left[\int_{0}^{T} \hat{\mu}_t^2dt \right]\\
&&+\frac{2\left(\frac{\beta}{\beta-1}-\frac{1}{2} \right)}{\left( 2\lambda\left(\beta-1 \right)\sigma_S^4\right)}\mathbb{C}\text{ov}\left[\hat{\mu}_T^2,\int_{0}^{T} \hat{\mu}_t^2dt \right].
\end{eqnarray*}
Moreover
\begin{eqnarray*}
\mathbb{V}\text{ar}\left[ \hat{\mu}_T^2 \right]&=&\mathbb{C}\text{ov}\left[\hat{\mu}_T^2,\hat{\mu}_T^2 \right],\\
\mathbb{V}\text{ar}\left[\int_{0}^{T} \hat{\mu}_t^2dt \right]&=&\int_{0}^{T}\int_{0}^{T}\mathbb{C}\text{ov}\left[\hat{\mu}_s^2,\hat{\mu}_t^2 \right]ds dt,\\
\mathbb{C}\text{ov}\left[\hat{\mu}_T^2,\int_{0}^{T} \hat{\mu}_t^2dt \right]&=&\int_{0}^{T}\mathbb{C}\text{ov}\left[\hat{\mu}_s^2,\hat{\mu}_T^2 \right]ds,
\end{eqnarray*}
and the expression of $\mathbb{C}\text{ov}\left[\hat{\mu}_s^2,\hat{\mu}_t^2 \right]$ is given in Lemma \ref{LemmaAutocovKalman} (see Appendix \hyperref[sec::A]{B}). Then 
\begin{eqnarray*}
\mathbb{V}\text{ar}\left[ \hat{\mu}_T^2 \right]&=&\frac{\lambda^2 \sigma_S^4 \left(\beta-1 \right)^4 }{2} \left( 1-2e^{-2\lambda T}+e^{-4\lambda T}\right) ,\\
\mathbb{V}\text{ar}\left[\int_{0}^{T} \hat{\mu}_t^2dt \right]&=&\frac{\lambda^2 \sigma_S^4 \left(\beta-1 \right)^4 }{2}\left(\frac{1-e^{-2\lambda T}}{2\lambda}\right. \\
&&\left. +\frac{e^{-2\lambda T}-e^{-4\lambda T}}{2\lambda} -2Te^{-2\lambda T}\right)  ,\\
\mathbb{C}\text{ov}\left[\hat{\mu}_T^2,\int_{0}^{T} \hat{\mu}_t^2dt \right]&=&\frac{\lambda \sigma_S^4 \left(\beta-1 \right)^4 }{2}\left(T-\frac{5}{4\lambda}+\frac{e^{-2\lambda T}}{\lambda}\right. \\
&&\left. +\frac{e^{-4\lambda T}}{4\lambda T}+2Te^{-2\lambda T} \right) .
\end{eqnarray*}  
Finally, using these expressions and tending $T$ to $\infty$, Equations (\ref{AsymptoticVariance}) and (\ref{AsymptoticSharpeRatio}) follow.
\end{proof}

The following result is a corollary of the previous theorem. It represents the asymptotic expectation, variance and Sharpe ratio of the logarithmic return as a function of the signal-to-noise-ratio and of the trend mean reversion speed $\lambda$.
\begin{cor}\label{corOpt}
Consider the portfolio given by Equation  (\ref{SelfFinancing}). In this case:
\begin{eqnarray}
\lim_{T \rightarrow \infty} \frac{\mathbb{E}\left[ \ln\left(  \frac{P_{T}}{P_{0}} \right)  \right] }{T}&=& \frac{1}{2}\left(\text{SNR}+\lambda-\sqrt{\lambda\left(\lambda+2\text{SNR} \right) }\right)\label{AsymptoticExpSNR},\\
\lim_{T \rightarrow \infty} \frac{\mathbb{V}\text{ar}\left[ \ln\left(  \frac{P_{T}}{P_{0}} \right)  \right] }{T}&=&\frac{\text{SNR}^2}{2 \lambda}\label{AsymptoticVarSNR},\\
\lim_{T \rightarrow \infty} \text{SR}_T&=&\left( \frac{\lambda}{2}\right) ^{3/2}\frac{\left(\sqrt{1+\frac{2\text{SNR}}{\lambda}}-1 \right)^2 }{\text{SNR}}\label{AsymptoticSharpeRatioSNR},
\end{eqnarray}
where $\text{SNR}$ is the signal-to-noise-ratio (see Equation (\ref{SNR})).
Moreover:
\begin{enumerate}
\item For a fixed parameter value $\lambda$, 
\subitem - the asymptotic expected logarithmic return is an 
\subitem \ \ increasing function of $\text{SNR}$,
\subitem - the asymptotic Sharpe ratio is an increasing function of 
\subitem \ \ $\text{SNR}$.
\item For a fixed parameter value $\text{SNR}$, 
\subitem - the asymptotic expected logarithmic return is a decreasing 
\subitem \ \ function of  $\lambda$,
\subitem - the asymptotic Sharpe ratio is a decreasing function of  $\lambda$ 
\subitem \ \ if:
\begin{eqnarray}\label{DecreasingSharpeCondition}
\text{SNR}<\frac{3}{2}\lambda,
\end{eqnarray}
\subitem \ \ and an increasing function of $\lambda$ if $\text{SNR}>\frac{3}{2}\lambda$.
\end{enumerate}
The maximum asymptotic Sharpe ratio is attained for $\lambda=\frac{2}{3}\text{SNR}$ and is equal to:
\begin{eqnarray}\label{SharpeRatioMax}
\text{SR}_\infty^\text{Max}=\frac{\sqrt{\text{SNR}}}{3^{3/2}}.
\end{eqnarray}
\end{cor}
\begin{proof}
Using Equation (\ref{SNR}) and Equation (\ref{BetaDef}), it follows that:
\begin{eqnarray*}
\beta=\sqrt{1+\frac{2\text{SNR}}{\lambda}}.
\end{eqnarray*}
Injecting this expression in Equation (\ref{AsymptoticExp}), we find:
\begin{eqnarray*}
\lim_{T \rightarrow \infty} \frac{\mathbb{E}\left[ \ln\left(  \frac{P_{T}}{P_{0}} \right)  \right] }{T}=L\left( \text{SNR},\lambda\right),
\end{eqnarray*}
where
\begin{eqnarray*}
L\left( \text{SNR},\lambda\right)=\frac{1}{2}\left(\text{SNR}+\lambda-\sqrt{\lambda\left(\lambda+2\text{SNR} \right) }\right).
\end{eqnarray*}
Since
\begin{eqnarray*}
\frac{\partial L\left( \text{SNR},\lambda\right)}{\partial\text{SNR}}=\frac{1}{2}\left(1-\frac{1}{\sqrt{1+\frac{2\text{SNR}}{\lambda}}} \right) \geq 0,
\end{eqnarray*}
the asymptotic expected logarithmic return is an increasing function of $\text{SNR}$. Moreover:
\begin{eqnarray*}
\frac{\partial L\left( \text{SNR},\lambda\right)}{\partial\lambda}=\frac{1}{2}\left(1-\frac{\lambda+\text{SNR}}{\sqrt{\lambda\left(\lambda +2\text{SNR}\right) }} \right) \leq 0,
\end{eqnarray*}
it follows that the asymptotic expected logarithmic return is a decreasing function of $\lambda$.
Moreover, using Equations (\ref{SNR}), (\ref{BetaDef}) and  (\ref{AsymptoticVariance}), Equation (\ref{AsymptoticVarSNR}) follows. 

Now, with Equations (\ref{SNR}), (\ref{BetaDef}) and  (\ref{AsymptoticSharpeRatio}), we find:
\begin{eqnarray*}
\lim_{T \rightarrow \infty} \text{SR}_T= \text{SR}_\infty\left(\text{SNR},\lambda \right),
\end{eqnarray*}
where
\begin{eqnarray*}
\text{SR}_\infty\left(\text{SNR},\lambda \right)=\left( \frac{\lambda}{2}\right) ^{3/2}\frac{\left(\sqrt{1+\frac{2\text{SNR}}{\lambda}}-1 \right)^2 }{\text{SNR}}.
\end{eqnarray*}
Since
\begin{eqnarray*}
\frac{\partial \text{SR}_\infty }{\partial\text{SNR}}=\frac{\lambda^{\frac{3}{2}}\left(\sqrt{1+\frac{2\text{SNR}}{\lambda}}-1 \right)^2 }{2\text{SNR}^2\sqrt{2\left(1+\frac{2\text{SNR}}{\lambda} \right) }}\geq 0,
\end{eqnarray*}
the asymptotic Sharpe ratio is an increasing function of $\text{SNR}$. Moreover:
\begin{eqnarray*}
\frac{\partial \text{SR}_\infty}{\partial\lambda}=\frac{\left(\sqrt{1+\frac{2\text{SNR}}{\lambda}}-1 \right) \left(3\lambda \left(1- \sqrt{1+\frac{2\text{SNR}}{\lambda}}\right) +2\text{SNR} \right)}{4\sqrt{2\lambda}\text{SNR}\sqrt{1+\frac{2\text{SNR}}{\lambda}}}.
\end{eqnarray*}
Then, the sign of $\frac{\partial \text{SR}_\infty}{\partial\lambda}$ is given by the sign of:
\begin{eqnarray*}
\text{A}\left(\text{SNR},\lambda \right)= \left(3\lambda \left(1- \sqrt{1+\frac{2\text{SNR}}{\lambda}}\right) +2\text{SNR} \right).
\end{eqnarray*}
Using $\beta=\sqrt{1+\frac{2\text{SNR}}{\lambda}}$, this expression can be factorised:
\begin{eqnarray*}
\text{A}\left(\text{SNR},\lambda \right)= \lambda\left(\beta-1 \right) \left(\beta-2 \right).
\end{eqnarray*}
Since $\beta \geq 1$, $\text{A}\left(\text{SNR},\lambda \right)$ is negative if and only if $\beta\leq 2$ (and positive if and only if $\beta \geq 2$), which is equivalent to the condition of Equation (\ref{DecreasingSharpeCondition}). Equation (\ref{SharpeRatioMax}) is obtained using $\text{SNR}=\frac{3}{2}\lambda$ in Equation (\ref{AsymptoticSharpeRatioSNR}). Note that $\text{SR}_\infty$ is always positive. Since $\text{SR}_\infty$ is an increasing function of $\lambda$ if  $\lambda<\frac{2}{3}\text{SNR}$ and a decreasing function after this point, the maximum value of this function is given by Equation (\ref{SharpeRatioMax}).
\end{proof}
\subsection{Impact of partial information on the optimal strategy} 
In order to measure the impact of the investor's inability to observe the trend on the optimal strategy performance, we introduce the partial information factor. This indicator represents the ratio  between the asymptotic Sharpe ratio of the optimal strategy with partial information and the asymptotic Sharpe ratio of the optimal strategy with complete information:
\begin{equation}\label{PIFDef}
\text{PIF}=\frac{\text{SR}_\infty}{\text{SR}_\infty^\text{o}},
\end{equation}
where $\text{SR}_\infty$ is the asymptotic Sharpe ratio of the optimal strategy with partial information, and $\text{SR}_\infty^\text{o}$ is the asymptotic Sharpe ratio of optimal strategy with full information. 
The following theorem gives the analytic form of this indicator.
\begin{theo}
The partial information factor is given by:
\begin{eqnarray}\label{PIF}
\text{PIF}=\left( \frac{\lambda}{\text{SNR}}\right) ^{3/2}\frac{\left(\sqrt{1+\frac{2\text{SNR}}{\lambda}}-1 \right)^2 }{\sqrt{2}},
\end{eqnarray}
where $\text{SNR}$ is the signal-to-noise-ratio (see Equation (\ref{SNR})). 

If $\text{SNR}<\frac{3}{2}\lambda$ (respectively, $\text{SNR}>\frac{3}{2}\lambda$): 
\begin{enumerate}
\item For a fixed parameter value $\text{SNR}$, this indicator is a decreasing function (respectively, an increasing function)  of  $\lambda$.
\item For a fixed parameter value $\lambda$, this indicator is an increasing function (respectively, a decreasing function) of $\text{SNR}$.
\end{enumerate}
Moreover:
\begin{equation}\label{PIFMAX}
\text{PIF} \leq \frac{2}{3^{3/2}},
\end{equation}
and this bound is attained for $\lambda=\frac{2}{3}\text{SNR}$.
\end{theo}
\begin{proof}
This expression of the partial information factor is a consequence of Equations (\ref{SharpeRatioOmniscient}) and  (\ref{AsymptoticSharpeRatioSNR}). Moreover:
\begin{eqnarray*}
\frac{\partial \text{PIF}}{\partial\text{SNR}}=\frac{\sqrt{\lambda } \left(\sqrt{\frac{2 \text{SNR}}{\lambda }+1}-1\right) \left(3 \lambda  \left(\sqrt{\frac{2 \text{SNR}}{\lambda }+1}-1\right)-2 \text{SNR}\right)}{2 \text{SNR}^{5/2} \sqrt{\frac{4 \text{SNR}}{\lambda }+2}}.
\end{eqnarray*}
This expression is positive if and only if $\text{SNR}\leq\frac{3}{2}\lambda$. The dependency on the mean reversion speed $\lambda$ comes from Corollary \ref{corOpt}.
\end{proof}
\begin{remarque}
Equation (\ref{PIFMAX}) shows that in the best configuration (with $\lambda=\frac{2}{3}\text{SNR}$), the asymptotic Sharpe ratio of the optimal strategy with partial information is approximatively equal to $38.49\%$ of the asymptotic Sharpe ratio of the optimal strategy with complete information. 

Moreover, the intuition tells us that a high signal-to-noise ratio and a small trend mean reversion speed $\lambda$ involves a small impact of partial information on the optimal strategy performance (and then a high $\text{PIF}$). This is true if and only if $\text{SNR}\leq\frac{3}{2}\lambda$.
\end{remarque}
\section{Simulations}
In this section, numerical examples are computed in order to illustrate the analytical results of the previous sections. The figure \hyperref[Figu1]{1}  represents the asymptotic Sharpe ratio of the optimal strategy with full information as a function of the signal-to-noise ratio. If the signal-to-noise ratio is inferior to 1, which corresponds to a trend standard deviation inferior to the volatility of the risky asset, the asymptotic Sharpe ratio of the optimal strategy with complete information is inferior to 0.5. 

Now, suppose that $\lambda \in \left[1,252\right]$ and that the trend is an unobservable process. The figure \hyperref[Figu2]{2} represents the asymptotic Sharpe ratio of the optimal strategy with partial information as a function of the trend mean reversion speed $\lambda$ and of the signal-to-noise ratio. Since $\lambda \in \left[1,252\right]$ and SNR$<1$, Equation (\ref{DecreasingSharpeCondition}) is satisfied and this Sharpe ratio is an increasing function of SNR and a decreasing function of $\lambda$. Moreover, the maximal value is inferior to 0.2. We also observe that, even with a high signal-to-noise ratio, a high mean reversion parameter $\lambda$ leads to a small Sharpe ratio.  

The figure \hyperref[Figu3]{3} represents the partial information factor, which corresponds to the ratio between the asymptotic Sharpe ratios of the optimal strategy with partial  and full  information (see Equation (\ref{PIFDef})). Using Equation (\ref{PIFMAX}), this indicator is bound by $\frac{2}{3^{3/2}}$. Since $\text{SNR}<\frac{3}{2}\lambda$, this indicator is a decreasing function of $\lambda$ and an increasing function of SNR. Even with a high signal-to-noise ratio, a high mean reversion parameter $\lambda$ leads
to a negligible performance of the optimal strategy with partial information compared to the case with full information.

The figures \hyperref[Figu4]{4} and \hyperref[Figu5]{5} represents the asymptotic Sharpe ratio of the optimal strategy with partial information and the partial information factor as functions of the signal-to-noise ratio and of $\lambda$ with $\lambda \in \left[0,2\right]$. Theses figures illustrate that, if  
$\text{SNR}>\frac{3}{2}\lambda$, these quantities are increasing functions of the trend mean reversion  speed $\lambda$ (and the partial information factor is also a decreasing function of the signal-to-noise ratio).
\begin{figure}[H]
\begin{center}
   \includegraphics[totalheight=6cm]{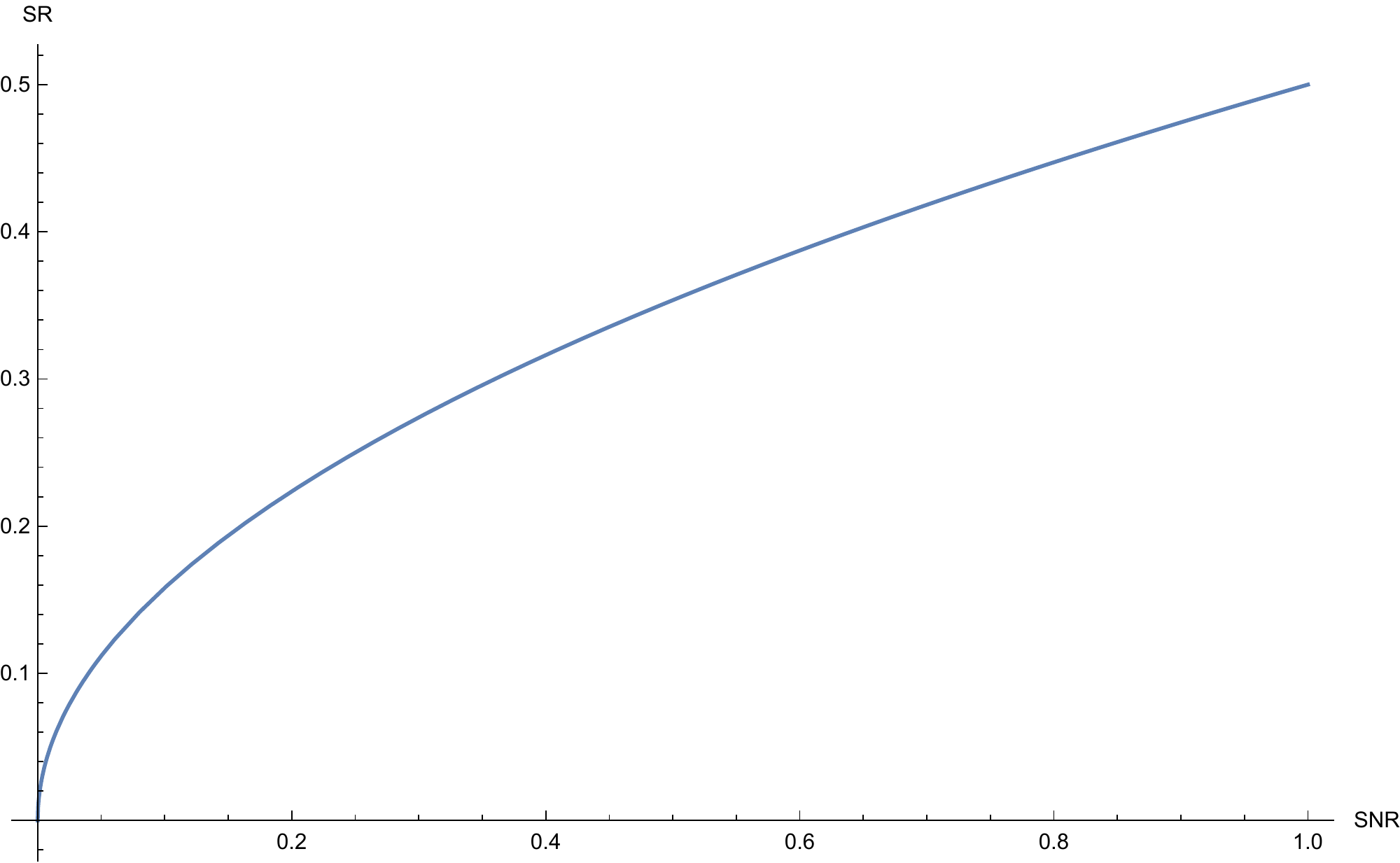}   
  \caption{Asymptotic Sharpe ratio of the optimal strategy with complete information as a function of the signal-to-noise ratio.}
\end{center}\label{Figu1}
\end{figure}
\begin{figure}[H]
\begin{center}
   \includegraphics[totalheight=7cm]{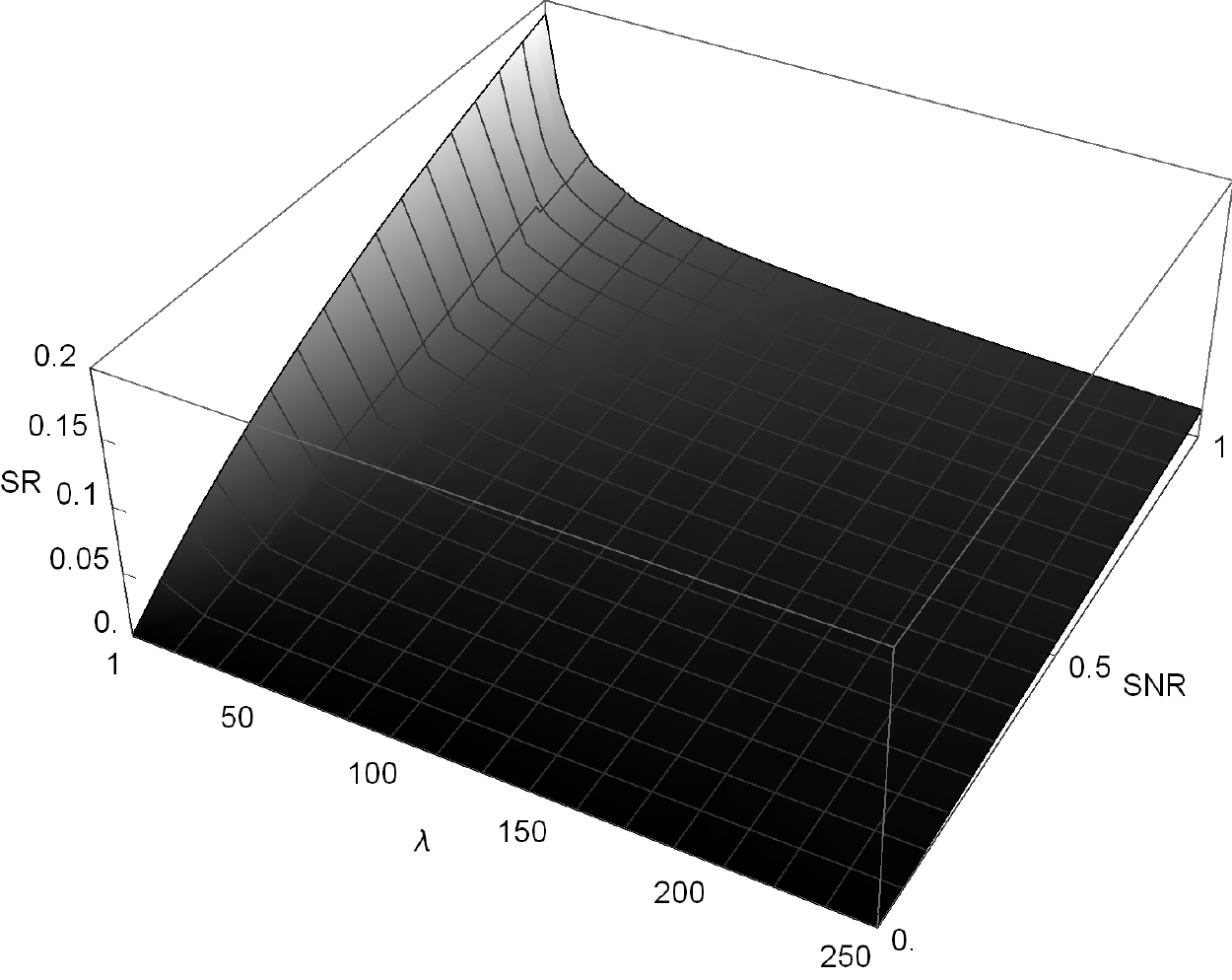}   
  \caption{Asymptotic Sharpe ratio of the optimal strategy with partial information as a function of the trend mean reversion speed $\lambda$ and of the signal-to-noise ratio.}
\end{center}\label{Figu2}
\end{figure}
\begin{figure}[H]
\begin{center}
   \includegraphics[totalheight=7cm]{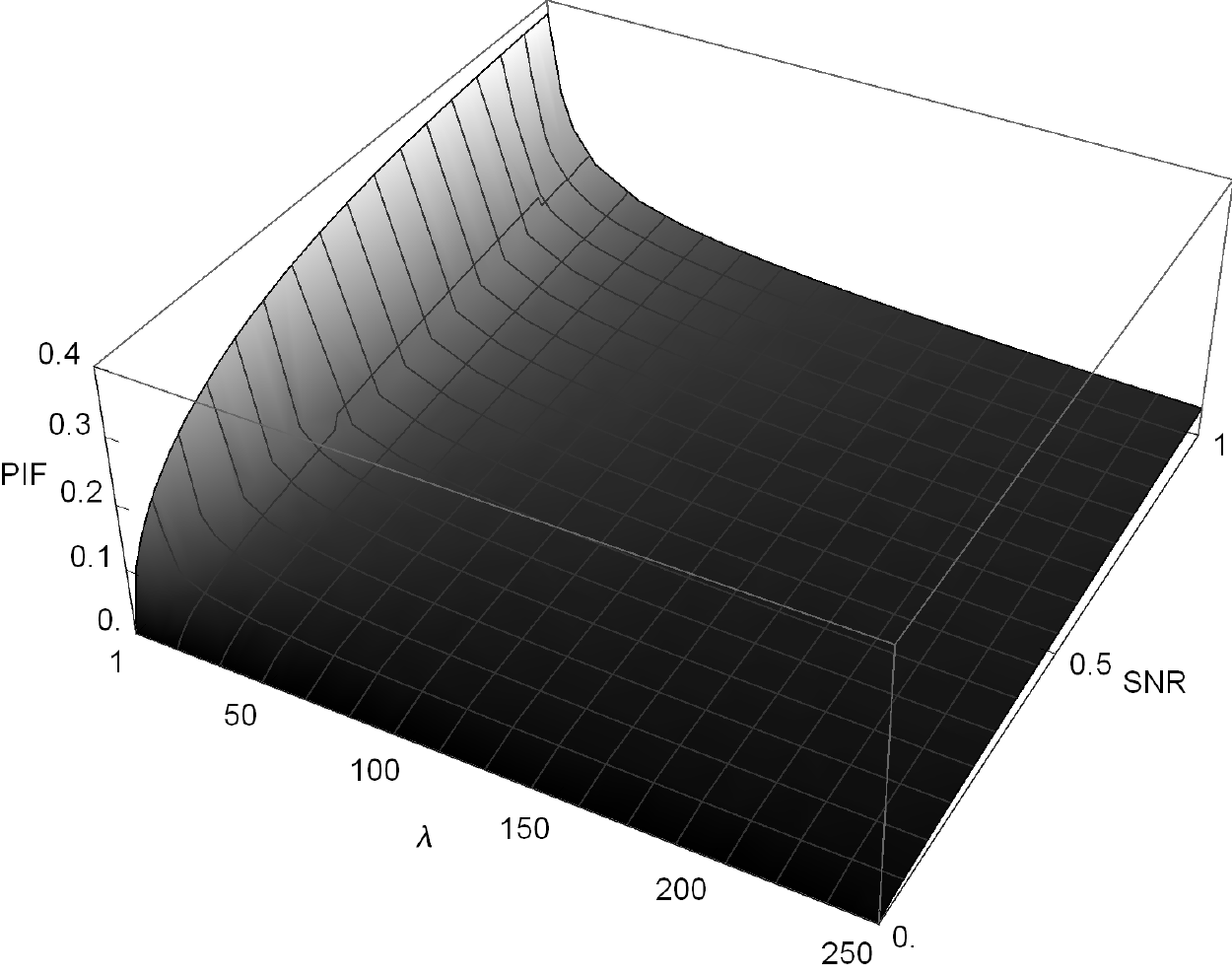}   
  \caption{Partial information factor as a function of the trend mean reversion speed $\lambda$ and of the signal-to-noise ratio.}
\end{center}\label{Figu3}
\end{figure}

\begin{figure}[H]
\begin{center}
   \includegraphics[totalheight=5.5cm]{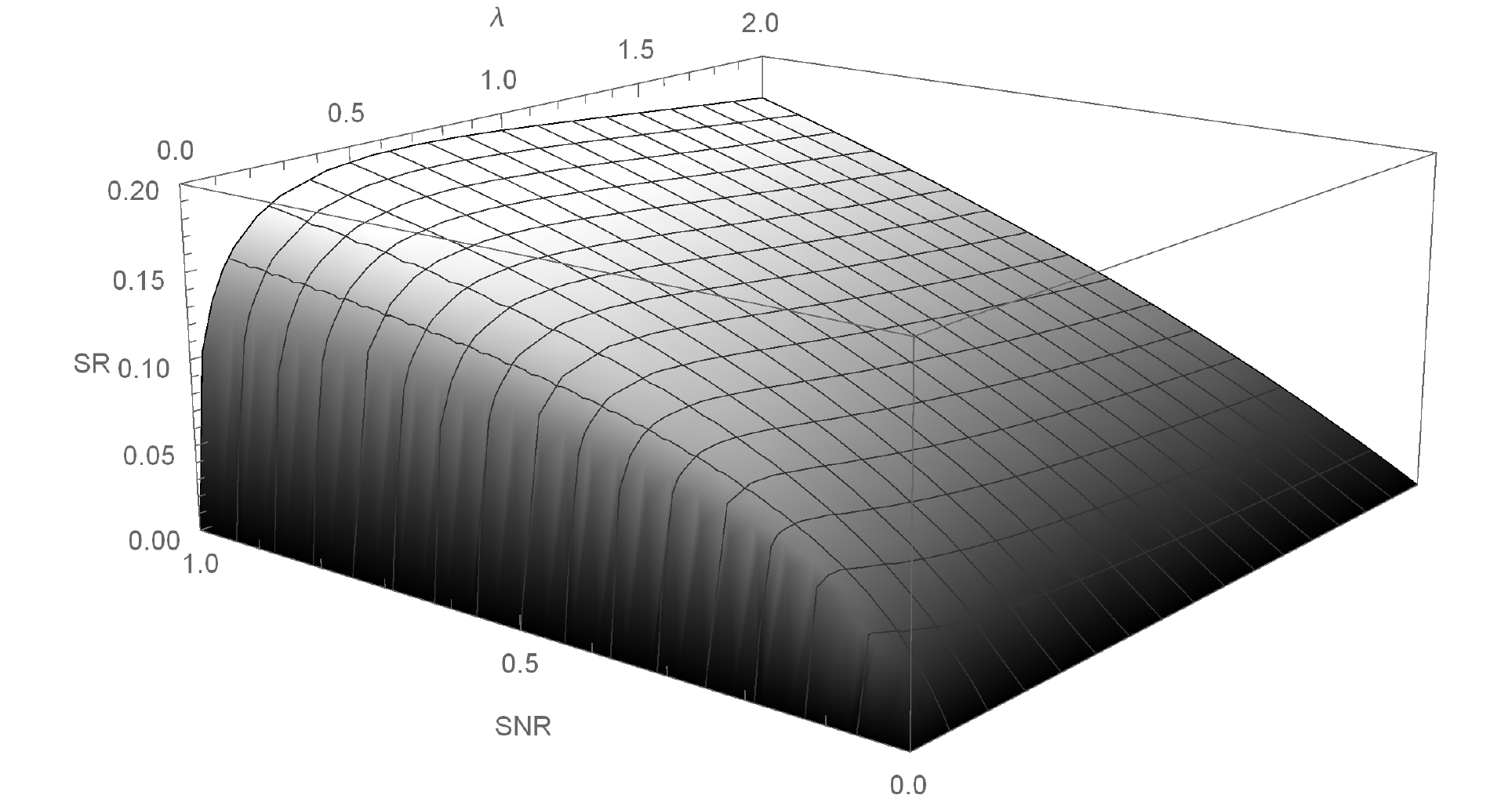}   
  \caption{Asymptotic Sharpe ratio of the optimal strategy with partial information as a function of the trend mean reversion speed $\lambda$ and of the signal-to-noise ratio with $\lambda \in \left[0,2\right]$.}
\end{center}\label{Figu4}
\end{figure}

\begin{figure}[H]
\begin{center}
   \includegraphics[totalheight=5.5cm]{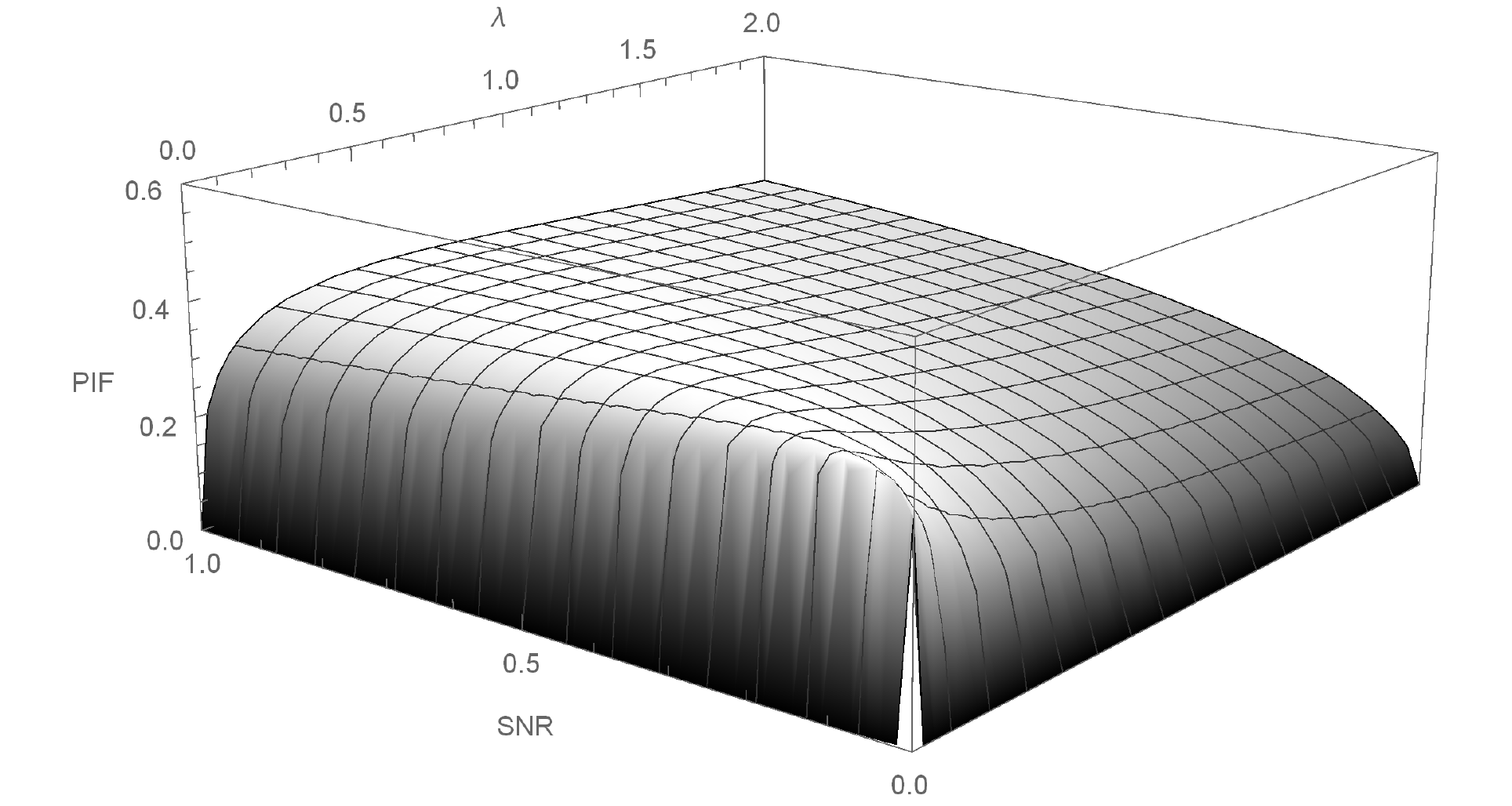}   
  \caption{Partial information factor as a function of the trend mean reversion speed $\lambda$ and of the signal-to-noise ratio with $\lambda \in \left[0,2\right]$.}
\end{center}\label{Figu5}
\end{figure}
\section{Conclusion}
The present work quantifies the loss of performance in the optimal trading strategy due to partial information with a model based on an unobserved mean-reverting diffusion. 

If the trend is observable, we show that the asymptotic Sharpe ratio of the optimal strategy is only an increasing function of the signal-to-noise ratio. 

Under partial information, this asymptotic Sharpe ratio becomes a function of the signal-to-noise ratio and of the trend mean reversion speed. Even if the asymptotic Sharpe ratio is also an increasing function of the signal-to-noise ratio, we find that the dependency on the trend mean reversion speed is not monotonic. Indeed, this is an unimodal (increasing then decreasing) function of the trend mean reversion speed.

We also show that the ratio between the asymptotic Sharpe ratio of the optimal strategy with partial information and the asymptotic Sharpe ratio of the  optimal strategy with complete information is bounded by a threshold equal to $\frac{2}{3^{3/2}}$. Given this result, we surely conclude that the impact of partial information on the optimal strategy is not negligible.

Moreover, the simulations show that even with a high signal-to-noise ratio, a high trend mean reversion speed leads to a negligible performance of the optimal strategy under partial information compared to the performance of the  optimal strategy with complete information.

\newpage
\section*{Appendix A: Proof of Proposition \ref{ObservableFramework}}
\label{sec::ProofObservableFramework}
\begin{proof}
Let $K$ be a $\left( \mathbb{P} ,\left\lbrace  \mathcal{F}_{t} \right\rbrace\right) $ martingale defined by:
\begin{eqnarray*}
\frac{dK_t}{K_t}=\frac{-\mu_t}{\sigma_{S}}dW_t^{S},
\end{eqnarray*}
and the probability measure $\widetilde{\mathbb{P}}$ defined by:
\begin{eqnarray*}
\frac{d\widetilde{\mathbb{P}}}{d\mathbb{P}}=K_T.
\end{eqnarray*}
With the Girsanov's theorem, it follows that the process:
\begin{eqnarray*}
\widetilde{W}_t^{S}=W_t^{S} + \int_{0}^{t}\frac{\mu_s}{\sigma_{S}} ds,
\end{eqnarray*}
is a $\left( \widetilde{\mathbb{P}} ,\left\lbrace  \mathcal{F}_{t} \right\rbrace\right) $ Wiener process. Note also that:
\begin{eqnarray*}
\frac{dS_t}{S_t}=\sigma d \widetilde{W}_t^{S}.
\end{eqnarray*}
Now, introduce the process $N$, defined by:
\begin{eqnarray*}
N_t=\widetilde{W}_t^{S}- \int_{0}^{t}\frac{E\left[  \mu_s | \mathcal{F}^S_s\right] }{\sigma_{S}} ds,
\end{eqnarray*}
as $\widetilde{W}_t^{S}$ and $\mathbb{E}\left[ \mu_{t} | \mathbb{F}^{S}_{t} \right]$ are $\left\lbrace \mathcal{F}^{S}_{t} \right\rbrace$ measurable, $N_t$ is $\left\lbrace \mathcal{F}^{S}_{t} \right\rbrace$ measurable. The process $N$ is also integrable.
Let $\tau$ be a bounded stopping time. we have
\begin{equation*}
       \mathbb{E}\left[ N_{\tau} \right]=\mathbb{E}\left[ N_{0} \right]=0.
\end{equation*}
Then, $N$ is a continuous martingale and $N_{0}=0$. Note that
\begin{equation*}
       d \left\langle N \right\rangle _{t}=dt
\end{equation*}
Using Levy's criteria, the process $N$ is a $\left( \mathbb{P} ,\left\lbrace  \mathcal{F}_{t}^{S} \right\rbrace\right)$ Wiener process.
\end{proof}

\newpage
\section*{Appendix B: Auto-covariance function of the square steady state Kalman filter}
\label{sec::A}
the following lemma gives the auto-covariance function of the process $\hat{\mu}^2$:
\begin{lemme}\label{LemmaAutocovKalman}
Consider the process $\hat{\mu}$ defined in Equation (\ref{ContinuousEstimate}). Its auto-covariance function is given by:
\begin{eqnarray}\label{AutocovKalman}
\mathbb{C}\text{ov}\left[\hat{\mu}_s^2,\hat{\mu}_t^2 \right]=\frac{\lambda^2\sigma_S^4\left(\beta-1 \right)^4}{2} e^{-2\lambda t}\left(e^{2\lambda s}+e^{-2\lambda s}-2 \right)  ,
\end{eqnarray}
with $0\leq s\leq t$.
\end{lemme}
\begin{proof}
Since $\hat{\mu}$ is a centred Ornstein Uhlenbeck process, there exists a Brownian motion $B$ such that, for all $s \in \mathbb{R}_+$:
\begin{eqnarray*}
\hat{\mu}_s=e^{-\lambda s}\lambda\sigma_S\left(\beta-1 \right)B_{f\left(s \right) }, 
\end{eqnarray*}
where $f\left(s \right)=\frac{e^{2\lambda s}-1}{2\lambda}$ is a time change. Then, for all $s,t$ such that $0\leq s\leq t$, we have:
\begin{eqnarray*}
\mathbb{C}\text{ov}\left[\hat{\mu}_s^2,\hat{\mu}_t^2 \right]=e^{-2\lambda \left( t+s\right)  }\lambda^4\sigma_S^4\left(\beta-1 \right)^4\mathbb{C}\text{ov}\left[B_{f\left(s \right)}^2,B_{f\left(t \right)}^2\right].
\end{eqnarray*}
Since $B$ is a Wiener process:
\begin{eqnarray*}
\mathbb{E}\left[B_{f\left(s \right)}^2\right]=f\left(s \right).
\end{eqnarray*}
Let $\left\lbrace  \mathcal{F}^B_{t} \right\rbrace$ be the filtration generated by the process $B$. So:
\begin{eqnarray*}
\mathbb{E}\left[B_{f\left(s \right)}^2,B_{f\left(t \right)}^2\right]&=&\mathbb{E}\left[B_{f\left(s \right)}^2 \mathbb{E}\left[B_{f\left(t \right)}^2 | \mathcal{F}^B_{s}\right] \right]\\
&=& \scriptstyle \mathbb{E}\left[B_{f\left(s \right)}^2 \mathbb{E}\left[\left( B_{f\left(s \right)}^2+2\int_{f\left(s \right)}^{f\left(t \right)}B_u dB_u + f\left(t \right)-f\left(s \right)\right)  | \mathcal{F}^B_{s}\right] \right]\\
&=&\mathbb{E}\left[B_{f\left(s \right)}^2 \left(B_{f\left(s \right)}^2 + f\left(t \right)-f\left(s \right) \right) \right]\\
&=& 3f\left(s \right)^2+\left(f\left(t \right)-f\left(s \right) \right) f\left(s \right).
\end{eqnarray*}
Then 
\begin{eqnarray*}
\mathbb{E}\left[B_{f\left(s \right)}^2,B_{f\left(t \right)}^2\right]=2f\left(s \right)^2+f\left(t \right)f\left(s \right),
\end{eqnarray*}
and Equation (\ref{AutocovKalman}) follows using:
\begin{eqnarray*}
\mathbb{C}\text{ov}\left[B_{f\left(s \right)}^2,B_{f\left(t \right)}^2\right] = \mathbb{E}\left[B_{f\left(s \right)}^2,B_{f\left(t \right)}^2\right]-\mathbb{E}\left[B_{f\left(s \right)}^2\right]\mathbb{E}\left[B_{f\left(t \right)}^2\right].
\end{eqnarray*}
\end{proof}
\newpage
\vspace{60pt}
\bibliographystyle{alpha}
\bibliography{bibli}
\end{document}